\documentclass{article}

\usepackage{graphicx}
\usepackage{xcolor}
\usepackage{multirow}

\usepackage{amsmath,amssymb,amsthm}
\usepackage{thmtools}

\theoremstyle{plain}
\newtheorem{theorem}{Theorem}
\newtheorem{lemma}{Lemma}
\newtheorem{proposition}{Proposition}

\theoremstyle{definition}

\newtheorem{definition}{Definition}

\providecommand{\detail}[1]{#1}
\renewcommand{\detail}[1]{}

\begin{document}

\title{A no-go theorem for sequential and retro-causal hidden-variable
  theories based on computational complexity.}

\author{Doriano Brogioli}

\maketitle

\begin{abstract}
  The celebrated Bell's no-go theorem rules out the hidden-variable
  theories falling in the hypothesis of locality and causality, by
  requiring the theory to model the quantum correlation-at-a-distance
  phenomena. Here I develop an independent no-go theorem, by
  inspecting the ability of a theory to model quantum
  \emph{circuits}. If a theory is compatible with quantum mechanics,
  then the problems of solving its mathematical models must be as hard
  as calculating the output of quantum circuits, i.e., as hard as
  quantum computing. Rigorously, I provide complexity classes
  capturing the idea of sampling from sequential (causal) theories and
  from post-selection-based (retro-causal) theories; I show that these
  classes fail to cover the computational complexity of sampling from
  quantum circuits. The result is based on widely accepted conjectures
  on the superiority of quantum computers over classical ones.  The
  result represents a no-go theorem that rules out a large family of
  sequential and post-selection-based theories. I discuss the
  hypothesis of the no-go theorem and the possible ways to circumvent
  them. In particular, I discuss the Schulman model and its
  extensions, which is retro-causal and is able to model quantum
  correlation-at-a-distance phenomena: I provides clues suggesting
  that it escapes the hypothesis of the no-go theorem.
\end{abstract}

\tableofcontents


\section{Introduction}

At its beginning, quantum mechanics was met with hostility, partially
because of its probabilistic nature: everybody remembers the famous
Einstein's refusal to ``believe that God plays dice''. Actually, the
dice appear random only if we look at the final outcome, but, looking
at their detailed motion, they are perfectly deterministic. Scientists
discussed if something similar could explain the randomness of quantum
mechanics: some variables, that are hidden or hardly visible, evolving
deterministically but showing a random behaviour at a blurred view.

Years later, Bell's theorem proved that hidden-variable models are not
viable, under the hypothesis of locality and causality. The research
then focused on hidden-variable models that violate either of the two
hypothesis. For a modern discussion of these topics I refer the reader
to Ref.~\cite{wharton2020}. It is curious to notice that the
above-mentioned statement of Einstein, in its complete form, did not
only refer to dice, symbolizing the randomness, but also to telepathy,
symbolizing the violation of locality.

Nowadays, the probabilistic aspect of quantum mechanics is no more
raising so much criticism. Maybe, because we gradually got used to
it. Or, more likely, because more pragmatic, still unanswered
questions arose and more serious problem were found in quantum
mechanics. One of the open questions refers to the huge Hilber space
needed to describe quantum mechanical systems, e.g. 30 qubits require
around one billion complex numbers to be represented. Although a
quantum computer with 30 qubits has a huge computational power, it is
still not clear if it really \emph{contains} (in any sense) this huge
amount of information~\cite{aaronson_ten_challenges}. Hence the
question: do we really need this huge space, or there is an
alternative mathematical theory that operates on a smaller space (on a
smaller number of variables), giving the same results? Such a theory
would fall inside the old definition of hidden-variable theoris, hence
the interest in these theories in this paper. I remark that such
alternative theories would be, likely, still hard to calculate,
notwithstanding the smaller number of variables; this point will play
a key role in the following.

It must be noticed that literature also reports a different, somehow
complementary, use of the name ``hidden-variable theory'': a theory
that deterministically gives the instantaneous values of physical
quantities, once the quantum state (wave) function is known.  An
example is the Bohm's pilot wave model.  This meaning of
``hidden-variable theory'' is outside the scope of this paper. An
interesting analysis of such theories, based on computational
complexity, is reported in Ref.~\cite{aaronson2005bis}.

In this paper, I use the term ``theory'' in the sense of a set of
mathematical models, each of them tentatively modelling a physical
system (this meaning is mutuated from the jargon of field theory and
deviates from the use in philosophy of science). Since we dropped the
concerns about the probabilistic nature of a theory, we can accept
that the variables of the models (more or less hidden) are
probabilistically determined; this represents a deviation from the
original meaning and aim of hidden-variable theories.

Rather than formally defining what a hidden-variable theory is, I will
focus on the discrimination between valid theories, i.e. theories that
give results in agreement with quantum mechanics, from invalid
theories, which instead fail to do so.  Ideally, the analysis applies
to any kind of theories, including the traditional formulations of
quantum mechanics based on the quantum state (wave) function and its
reformulations, e.g. based on the path integrals.

As already mentioned, the Bell's no-go theorem was a fundamental
milestone in science. It is based on the observation of a puzzling
phenomena of quantum mechanics, the correlation at a distance, arising
e.g. in the Einstein-Podolsky-Rosen experiment.  The theorem states
that a hidden-variable theory is not able to model such phenomena,
under the hypothesis that the theory is local and causal. In this
paper, I derive an independent no-go theorem, by discussing the
ability of a hidden-variable theory to model quantum \emph{circuits}.
This analysis allows us to exploit the results from computational
complexity theory. Computational complexity theory has already been
used to explore the properties of variants of quantum
mechanics~\cite{aaronson2005}. In general, it is a valuable
mathematical method for investigating physical theories.

Intuitively, the idea behind the no-go theorem is that the
computational complexity of solving the mathematical models of the
theory must be at least hard as simulating quantum
computation. Rigorously, to each theory, I associate the computational
problem of sampling from its mathematical models; then I define some
computational complexity classes, namely {\bf SampP}, {\bf PostSampP},
and {\bf PostSampP*}, such that the theory is not valid if its
associated computational problem lies in one of those classes.

An insight into the meaning of these classes can be grasped by
inspecting the Bell's theorem hypothesis.  An implicit hypothesis is
that the hidden variables evolve with time such that the state at a
time is calculated from the states at previous times. This condition
is associated to causality, i.e. the time order of cause and
effect. In circuits, there are finite steps, corresponding to the
gates, thus such models will be called ``sequential'': each step can
be calculated from the previous one, starting from an initial
state. In the context of circuits, the locality appears as the
condition that each step, or gate, only works on a fixed number of
bits. Such a theory would be local and causal, and thus it would be
ruled out by Bell's theorem. The computational problem associated to
such theories lies in {\bf SampP}, thus confirming that they are
not viable, through a totally different approach.

In order to escape causality, as an alternative to sequential models,
``all-at-once'' models have been proposed. In such models, there is no
initial state and no way to sequentially solve the model. A particular
kind of such models are the retro-causal models, in which conditions
are imposed both at the beginning and at the end of the process. If
the evolution from starting to ending time is probabilistic, it is
possible to describe retro-causality in the form of post-selection:
starting from an initial state, the hidden variables evolve
probabilistically; at the end, the measurement setting imposes its
constraints, deciding to accept or reject the final state. With
additional hypothesis, the computational problems associated to
post-selection theories belong to {\bf PostSampP} or {\bf PostSampP}
computational complexity classes, thus these theories are not valid.

An interesting example of retro-causal model is the Schulman model,
which has been extended to reproduce interesing quantum phenomena
connected to entanglement~\cite{wharton2020}. We can wonder if these
few models can be extended to a full theory modelling all the quantum
circuits (and thus, likely, the whole quantum mechanics). I will
provide a generalized Schulman theory, along with two simplifications.
One of these theories is sequential; its associated computational
problem lies in {\bf SampP} and thus I prove that the theory is not
valid. The second, simplified theory is retro-causal. Its associated
computational problem is likely computayionally harder and lies in
{\bf PostSampP*}: it is thus harder, but still not hard enough to make
the theory valid. The full theory escapes the hypothesis of the no-go
theorem by requiring a limit for vanishing parameters. This condition
is actually analogous to the known requirement for fine
tuning~\cite{wood2015}.

It could appear surprising that, here, I am inviting the scientists to
develop models that are hard to calculate, while a practical goal would
be to make the calculation easier. However, we cannot build any
hidden-variable theory that is both easy to calculate and able to model
every quantum circuit: the uderlying idea is that quantum mechanics is
intrinsically hard to calculate, whatever is the used calculation
method, including possible future advancements. This idea is not
supported by theorems but relies on conjectures that are however
widely believed to be true.

Informally, the conjecture is that quantum computers are more powerful
than classical computers. The existing quantum computers are still
toys and do not really show the so-called ``quantum
supremacy''. However, it is believed that, at least in principle, they
can perform calculations that classical computers cannot do in a
practically short time, not even with advanced algorithms. For
example, quantum computers can crack the cryptographic codes used to
certify our cash cards, but we believe that those codes are perfectly
safe, because no classical computer can crack them in reasonable time
and quantum computers are still in their infancy. This reasoning does
not mathematically prove the conjectures on the hardness of simulating
quantum mechanics, but it shows that a part of our life relies on
them.

The paper uses the jargon of computational complexity theory.  For a
general introduction to the topic, I refer the reader to
Ref.~\cite{arora_barak}.


\section{Summary}

Section~\ref{sect:schulman:model} is aimed at giving an example of a
hidden-variable theory. I introduce the Schulman model, some of its
extensions, and a generalization representing a candidate
hidden-variable theory of quantum mechanics. The details of this
theory can be found in Ref.~\cite{wharton2020} and are not repeated
here.

Section~\ref{sect:definition} reports a formal definition of
``theory'', which is used to define the problems of generating samples
according to the models. The validity of the theory is then assessed
based on the hardness of the computational problem for a suitable
quantum computational complexity class. The analysis is then brought forward
to prove a no-go theorem in Sect.~\ref{sect:no:go}.

The no-go theorem is finally used in
Sect.~\ref{sect:no:go:post:selection:schulman} to analyze some
variants of the extended Schulman theory. Results similar to Bell's
theorem and the need for fine tuning are recovered, from this
computational complexity theory approch.

The Sect.~\ref{sect:conclusion} summarizes the findings and suggests
the possible ways to make hidden-variable models viable.

Appendix~\ref{sect:decision} discusses some complexity classes related
to post-selection and proves relations among them. These results are
then used in Appendix~\ref{sect:no:go:sampling:post:selection} to
prove some of the propositions discussed in the paper.


\section{Example of hidden-variable theory: the extended Schulman theory}

\label{sect:schulman:model}

\begin{figure}
  \includegraphics{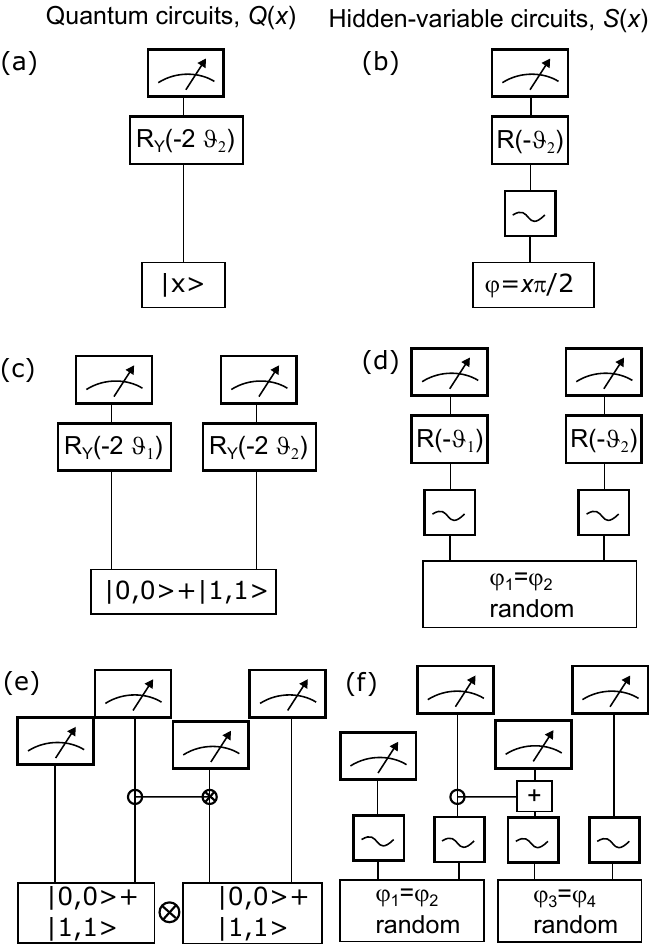}
  \caption{Some quantum circuits and the corresponding models from the
    extended Schulman theory. The direction is bottom up. Panels a and
    b: the input is a single bit $x$. A photon is prepared with a
    polarization angle $\vartheta_1=x \pi/2$. Its polarization is then
    measured at a different angle $\vartheta_2$. Panels c and d: two
    entangled photons are measured at different angles $\vartheta_1$
    and $\vartheta_2$. Panel e and f: an exmaple of more complex
    entanglement (the shown implementation of the CNOT gate is not general
    and only works in some cases).  Panels a, c, and e show the
    quantum circuit, panels b, d, and f the corresponding
    hidden-variable circuit. }
  \label{fig:schulman:schemes}
\end{figure}

\begin{figure}
  \includegraphics{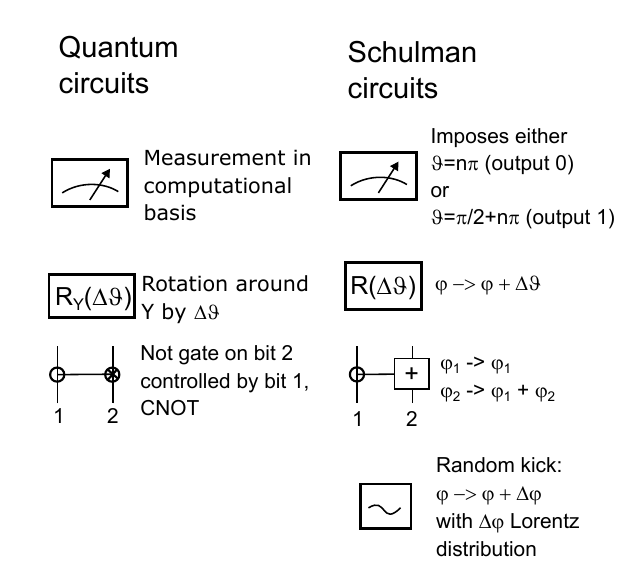}
  \caption{Legend for the symbols used for representing quantum
    circuits in Fig.~\ref{fig:schulman:schemes}.}
  \label{fig:schulman:legend}
\end{figure}

The aim of this section is to provide an example of hidden-variable
theory, which is known to work for some quantum systems.

\subsection{The Schulman model and some variants}

I first describe the Schulman model, a model that successfully
describes a quantum system and can be extended to some additional some
quantum systems.  These models are thoroughly discussed in
Ref.~\cite{wharton2020}.  It is still not clear if suitable extensions
can model every quantum phenomena or every quantum circuits. However,
these models are able to simulate some non-trivial quantum phenomena
and is considered as promising. Generalizing these models, we get a
candidate theory, whose validity will then be checked by means of the
techniques developed in this paper.

The starting point is the simulation of the passage of a polarized
photon through a polarizing beamsplitter. The experiment consists in
starting with a photon polarized at a given angle $\vartheta_1$. The
photon is sent through a polarizing beam splitter, set at an angle
$\vartheta_2$, such that the transmitted (resp., reflected) photons
are polarized at $\vartheta_2$ (resp., at
$\vartheta_2+\pi/2$). Quantum mechanics (actually, in this simple
case, classical optics suffices) gives the probability of having
transmission or reflection; the result is the Malus' law.  Two devices
finally detect the photons in the two branches. The Schulman model is
based on a hidden variable $\varphi$, set to $\varphi=\vartheta_1$ at
the beginning. The propagation of the photon is modelled as a sequence
of random kicks of amplitude $\Delta \varphi$, where $\Delta \varphi$
has a Lorentz distribution with width $\delta \varphi_L$:
\begin{equation}
  \mathcal{P}\left(\Delta \varphi\right)= \frac
          {\delta \varphi_L}
          {\delta \varphi_L^2 + \Delta \varphi^2}
  \label{eq:lorentz}
\end{equation}
The model works in the limit of vanishing width $\delta \varphi_L\to
0$ (this point is crucial and will be thoroughly discussed below). The
passage through the polarizing beam splitter is implemented with the
constraint that the final polarization of the photon, $\varphi$, must
be perfectly aligned with the polarizing beam splitter, i.e. either
$\varphi=\vartheta_2+n\pi$ and the photon is transmitted, or
$\varphi=\vartheta_2+\pi/2+n\pi$ and the photono is reflected. This
model gives the same prediction of the quantum mechanics, i.e. the
Malus law.

Although the evolution of $\varphi$ can be thought of as of
sequential, with a starting value and some sequential modification
events, the measurement setting impose a condition on the
outcoume. The model is defined as retro-causal, i.e. future conditions
influence the present.

Before discussing the variants of this model, I show how it can be
expressed in terms of quantum
circuits. Figure~\ref{fig:schulman:schemes}a graphycally shows a
quantum circuit that implements the above-described physical system.
The input of the quantum circuits defines the initial
polarization. Let us assume that the input $x$ a single bit,
representing either vertical or horizontal polarization.  The
calculation uses a single qubit, which starts in state $|x>$, where
$|0>$ (resp. $|1>$) represents a vertically (resp. horizontally)
polarized photon, $\vartheta_1=0$ (resp. $\vartheta_1=\pi/2$). The
polarizing beam splitter followed by the detectors is represented by
the ``measurement'' (see Fig.~\ref{fig:schulman:legend}). For better
schematizing the circuit, I assume that the measurement is always in
the computational basis, which means vertical vs. horizontal
polarization; the detection of a vertically (resp. horizontally)
polarized photon will be read as an output 0 (resp. 1). In order to
represent a measurement along a rotated direction, the rotation is
rather attributed to the photon itself, just before the measurement;
it is represented by the ``rotation'' gate $R_Y(-2 \vartheta_2)$ (see
Fig.~\ref{fig:schulman:legend}). Notice that the used representation
usually refers to spins, while here I only use the notation for
polarizations.

The Schulman model is graphycally shown in
Fig.~\ref{fig:schulman:schemes}b. The hidden variable is $\varphi$. Is
starts from $\varphi=\vartheta_1=x \pi/2$ (I remind that $x$ is the
single bit of input), i.e. it starts either as a vertically or
horizontally polarized photon. Then, during the propagation, the angle
$\varphi$ receives kicks with a Lorentz distribution. The Schulman
model prescribes several kicks during the propagation, but we already
know that only one of these kicks is statistically relevant, so we
only insert one of them in the circuit, with the ``kick'' symbol, a
box with a $\sim$ (see Fig.~\ref{fig:schulman:legend}). Then we have a
rotation by the angle $-\vartheta_2$. Finally, there is the
measurement; since it detects the vertical vs. horizontal
polarization, it imposes that either $\varphi=n\pi$ and the photon is
transmitted, or $\varphi=\pi/2+n\pi$ and the photono is
reflected. This hidden-variable circuit gives exactly the same
probability distribution of the quantum circuit of
Fig.~\ref{fig:schulman:schemes}a in the limit of vanishing width of
the Lorentz distribution of $\Delta \varphi$, which is simply the
Malus' law.

An extension of the Schulman model to a more complex case is shown in
Fig.~\ref{fig:schulman:schemes}c and d: it refers to two entangled
particles. The quantum circuit is shown in
Fig.~\ref{fig:schulman:schemes}c. It represents one of the versions of
the celebrated Einstein-Podolsky-Rosen thought experiment, where a
correlation at distance appears (althought it does not constitute a
communication of a signal). It is the core of the reasonings about the
Bell's theorem. An extension of the Schulman model, suitable for this
quantum circuit, is represented in Fig.~\ref{fig:schulman:schemes}c.
The polarization angle of the two photons are $\varphi_1$ and
$\varphi_2$, starting from random but equal values
$\varphi_1=\varphi_2$. It must be noticed that one of the two random
kicks could be removed, being statistically irrelevant.

Also in this case, the results of the quantum circuit and of its model
are the same. Hidden-variable theories are forbidden by Bell's theorem
under specific hypothesis, but, actually, the model
Fig.~\ref{fig:schulman:schemes}d works and generates the desired
correlations at distance, which are forbidden by Bell's theorem. This
is possible because this model is retro-causal, thus it violates one
of the hypothesis of Bell's theorem: the values of $\varphi_1$ and
$\varphi_2$ are constrained by the measurement settings, in this case,
by the $\vartheta_1$ and $\vartheta_2$ values, which enter in the
circuit after the random kicks.

A further example of extensions of Schulman model is reported in
Fig.~\ref{fig:schulman:schemes}e and f. It represents a circuit with
four qubits. Two couples of entangled quibits are generated, then they
are entangled by a CNOT gate. The quantum gate is represented, in the
model, by adding the polarization angle of the controlling photon to
the controlled one. This trick works in this specific case but is not
a general implementation of the CNOT gate.

\subsection{Generalizing: the extended Schulman theory}

\label{sect:schulman:theory}

The Schulman model can be extended to even more complex cases (not
reported here) but it has not been applied to a general description of
every quantum circuit. However, it is possible to generalize the
above-described models to get the ``extended Schulman theory'', which
contains all the possible models that we expect to get by extending
the Schulman model.

\begin{figure}
  \includegraphics{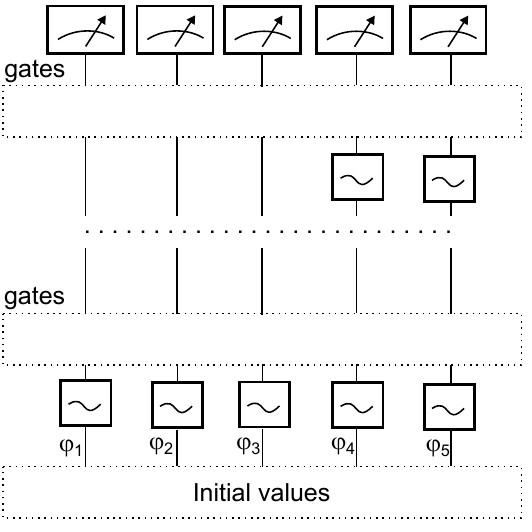}
  \caption{General scheme of the models belonging to the extended
    Schulman theory. The direction is bottom up.}
  \label{fig:general:schulman}
\end{figure}

A general scheme of the models belonging to the extended Schulman
theory is schematically shown in Fig.~\ref{fig:general:schulman}.

The first step is to assign the initial values of the hidden variables
$\varphi_j$; in the examples of Fig.~\ref{fig:schulman:schemes}, the
values can be either assigned to a given value or to a random value;
possibly, the same random value is used to set more than one variable.

After this step, there are two types of layers. One represents gates;
in the hidden-variable model, a gate is actually a deterministic
mapping between the input and the output. The other layer represents
the random kicks, according to the Lorentz distribution,
Eq.~\ref{eq:lorentz}; the random kicks can be applied to some signals
$\varphi_j$ or to all of them. These two layers are alternated in an
arbitrary number $m$. In analogy with the quantum circuits, I assume
that the number of such gates is polynomial in the number of signals
$\varphi_j$.

The last step consists in the measurement, which imposes a constraint
to the possible values of the angles $\varphi_j$: for a given $j$,
either $\varphi_j=\vartheta_j+n\pi$ or
$\varphi_j=\vartheta_j+\pi/2+n\pi$, for an integer $n$.

I will call this theory ``extended Schulman theory'', $T_{eS}$. It is
important to notice that the above definiton only provides the general
mathematical form of the models, without telling us which is the model
of a given quantum circuit. This notwithstanding, I will show that it
is possible to evaluate its validity, by the knowledge of the models,
even without knowing the association between models and quantum
circuits.


\section{Definition of ``theory'' and assessment of validity}

\label{sect:definition}

After giving an example of theory in Sect.~\ref{sect:schulman:model},
hiere I clarify and formally define the meaning of the term
``theory''. The meaning is quite generic: it includes quantum
mechanics itself, but it can even refer to a theory that only applies
to a small set of quantum circuit, or even none (a ``wrong''
theory). Based on the definition, I will formalize the concept of
validity, i.e. the ability of reproducing the results of quantum
mechanics.

\subsection{Formal definition of ``theory''}

A ``theory'' is a collection of models; each model has a
representation as a string and has an output of a fixed number of
bits. The output probability is defined by the model itself.  The
following definition captures the concept of ``theory'' in this sense.

\begin{definition}[Theory]
  A theory $T$ defines a set of models, with associated
  representations as strings. The set of representations of models is
  $S\subseteq\{0,1\}^*$. Each model $s\in S$ has a probabilistic
  output $y$ of size $N(s)$, $y\in \{0,1\}^{N(s)}$.  The output size
  is assumed to be smaller than the size of the representation of the
  model, $N(s)\le |s|$. The theory $T$ defines the probability
  distribution $\mathcal{P}_s(y)$ of the output $y$ of a model $s$.
\end{definition}

Notice that, at the stage of this definition, a theory does not
necessarily associate a model to a quantum circuit. This will be done
in a following step.

Quantum mechanics defines the outputs of quantum circuits. It
represents a first instance of theory, as defined as follows. 

\begin{definition}[Quantum mechanics $T_{QM}$]
  Quantum mechanics $T_{QM}$ is a theory, whose models are quantum
  circuits with their inputs (notice that I am limiting the scope to
  quantum circuits). Their representations are couples $s=(Q,x)$,
  where $Q$ is a representation of the circuit and $x$ is the input
  string. The set of representations is called $S_{QM}$. The
  representation $Q$ encodes the number of qubits $n$ and the number
  of gates $m$, assumed to be one- and two-qubit gates. The number of
  bits of output $N(s)$ equals the number of qubits. The output
  probability $\mathcal{P}^{QM}_{Q,x}(y)$ is defined according to
  quantum mechanics.
\end{definition}

I emphasize that the comparison of a theory $T$ with quantum mechanics
$T_{QM}$ is limited to quantum circuits. Although this limitation
leaves out several natural quantum phenomena, the comparison is
significant because most of the puzzling quantum phenomena can be
observed in suitable quantum circuits.

\subsection{What is a valid theory of quantum mechanics?}

It is now possible to discuss when a theory is valid, i.e. when it
gives the same results of of quantum mechanics, simply based on the
set of models. In order to be valid, we expect that each quantum
circuit $Q$ and input $x$ can be modelled in the theory by a
corresponding model $s$, such that $Q(x)$ and $s$ give the same output
probabilities. This idea is formalized by the following definitions.

\begin{definition}[Valid theory of quantum mechanics]
  \label{def:theory:covers}
  ``The theory $T$ is valid'' (or: is a valid theory of quantum
  mechanics) means that there is a mapping $R$ from the models
  $S_{QM}$ of $T_{QM}$ to the models $S$ of $T$, $R(Q,x)=s \in S$,
  such that the probability distribution of $T$ equals the probability
  distribution of $T_{QM}$:
  \begin{equation}
    \mathcal{P}^{QM}_{(Q,x)}(y) = \mathcal{P}_{R(Q,x)}(y)
  \end{equation}
  The mapping can be calculated in polynomial time in the size of the
  input.
\end{definition}

This definition includes the possibility that there are models of $T$
that do not correspond to any quantum circuit. A one-to-one mapping
would define a theory that is equivalent to quantum mechanics, however
this concept is outside the scope of this paper.

\subsection{Importance of the computational complexity of the mapping}

The mapping $R$ is a part of every hidden-variable theory, although it
is usually not expressed as a mathematical entity, but rather
described as a recipe for building a realization of the model matching
a given physical systems. However, in the literature on
hidden-variable theories, the properties of the reduction typically
remain implicit. Instead here, in Def.~\ref{def:theory:covers}, it is
requested that the mappings can be calculated in polynomial time in
the size of their inputs. This choice must be discuss here, because it
is fundamental for the plan of applying the computational complexity
theory.

For sure, the mapping $R$ must be computable. Moreover,
intuitively, is should be straightforward: no complex calculation
should be involved to devise $s$ from $(Q,x)$.

To see why the calculation of the mapping $R$ must be simple, imagine
a hidden-variable theory in which the model requires a long
calculation to be devised from the quantum circuit, e.g. including the
calculation of the quantum state (wave) function: likely, such a model
would be not deemed as convincing, even if its results are correct.

In the following, the mapping $R$ will be used as a reduction between
problems. A common requirement for a reduction in computational
complexity theory is that the reduction is done in polynomial time in
the input, usually when the involved complexity classes are closed
under polynomial-time reductions. This constraint is reasonable and
matches the intuitive idea of ``being straightforward''. In the
following subsection, the importance of the properties of the
reduction $R$ will become more clear. This explains the
request for a polynomial time reduction in Def.~\ref{def:theory:covers}.

The choice of a polynomial time reduction is admittedly arbitrary,
however it does not limit the validity of the results.  If we consider
the no-go theorems developed in this paper, they will include, as
hypothesis, the necessity of a polynomial time reduction. A
possibility to circumvent the no-go theorems is to develop a
hidden-variable theory in which the reduction, i.e. the correspondence
between the model and the physical systems, is impossible to calculate
than in polynomial time. This is a possibility, even if it is unlikely
that such a model can be considered credible.

\subsection{Definition of sampling problems}

In order to analyze the computational complexity, we must associate a
theory $T$ to a computational problem. First of all, we can chose
among strong or weak simulation~\cite{vandennest2010}: the former aims
at calculating the probabilities of events, the latter to sample the
output, i.e. to generate random samples with the requrested
probability.  I discuss the latter approach.

I start giving the definition of sampling problem, following
Ref.~\cite{aaronson2010}.

\begin{definition}[Sampling problem]
  \label{def:sampling:problem}
  A sampling problem is a collection of probability distributions
  $\mathcal{D}_x(y)$, for $x\in \{0,1\}^*$ and $y\in
  \{0,1\}^{p\left(|x|\right)}$, for some fixed polynomial $p(n)$.
\end{definition}

Informally, the problem is to sample from that distribution.  An
analogy with the probability distributions of the models $m$ of a
theory $T$ is immediately evident. We can thus associate a theory to a
sampling problem as follows.

\begin{definition}[Sampling problem: SAMP$(T)$]
  \label{def:problem:hidden:variable:sampling}
  The sampling problem SAMP$(T)$ is the sampling problem associated to
  a theory $T$.  The family of probability distributions of the
  sampling problem, $\mathcal{D}_x(y)$, corresponds to the probability
  distributions of the theory $T$, $\mathcal{P}_s(y)$, i.e. a
  probability distribution of the ourput $y$ for each model $s\in S$.
\end{definition}

The sampling problem associated to quantum circuits is SAMP$(T_{QM})$.

\subsection{Assessment of validity}

We can assess the validity of a theory $T$ comparing its
associated sampling problem, SAMP$(T)$, with the sampling problem
associated to quantum circuits, SAMP$(T_{QM})$. The idea is formalized
in the following proposition.

\begin{proposition}
  \label{prop:valid:reduce:samp:t:qm}
   If a theory $T$ is valid, then SAMP$(T_{QM})$ can be reduced to
   SAMP$(T)$.
\end{proposition}

\begin{proof}
  The proof relies on the identification of the reduction with the
  mapping between $T$ and $T_{QM}$.
\end{proof}

It is worth noting that the converse is not true. Indees, the
reducibility of SAMP$(T_{QM})$ to SAMP$(T)$ implies the presence of a
reduction $R$, which is also a mapping between $T_{QM}$ and $T$;
however, we cannot ensure that this mapping is meaningful as a
representation of quantum circuits or it is only a mathematical
correspondence.

An important class of sampling problems is {\bf SampBQP}. I give the
definition following Ref.~\cite{aaronson2010} with minor changes.

\begin{definition}[{\bf SampBQP} sampling class]
  \label{def:samp:bqp}
  The class {\bf SampBQP} is the class of sampling problems
  $\mathcal{D}_x(y)$, for $x\in \{0,1\}^*$ for which there exists a
  polynomial-time uniform family of quantum circuits $B_n$ that:
  \begin{enumerate}
  \item Takes an input $i=(x,0^{1/\epsilon})$;
  \item The circuit $B_{|i|}(i)$ returns a string $y\in
    \{0,1\}^{p\left(|x|\right)}$;
  \item The probability of getting a $y$ on a given input $x$,
    $\mathcal{C}_x(y)$, is close to $\mathcal{D}_x(y)$ as:
    \begin{equation}
      \sum_{y\in \{0,1\}^{p\left(|x|\right)}}
      \left| \mathcal{C}_x(y) - \mathcal{D}_x(y) \right| \le \epsilon
    \end{equation}
  \end{enumerate}
\end{definition}

Informally, this class approximates a probability distribution with
the output of a quantum circuit, thus, the problem of sampling from
quantum circuits, i.e. the models of $T_{QM}$, is complete for
it. We can thus rephrase the necessary condition of validity in terms of
hardness of the sampling problem for the class {\bf SampBQP}.

\begin{proposition}
  If a theory $T$ is valid (it is a valid theory of quantum
  mechanics), then SAMP$(T)$ is hard for {\bf SampBQP}.
\end{proposition}

\begin{proof}
  According to Def.~\ref{prop:valid:reduce:samp:t:qm}, since the
  theory $T$ is valid, then SAMP$(T_{QM})$ can be reduced to
  SAMP$(T)$.  It is then enough to show that SAMP$(T_{QM})$ is hard
  for {\bf SampBQP} (actually, it is even complete). This is proven by
  noticing that the class {\bf SampBQP} associates to each sampling
  problem $\mathcal{D}_x(y)$ a sampling problem $\mathcal{C}_x(y)$,
  corresponding to the sampling of the family $B_n$ of circuits. This
  latter problem is a subset of SAMP$(T_{QM})$.
\end{proof}

The results above state that it is possible to rule out a theory by
analyzing its associated sampling problem and determining its hardness
for {\bf SampBQP}. A theory with a too weak associated sampling
problem is for sure not valid.

It is worth remarking that this analysis is done without the need of
knowing the association between quantum circuits and the models of the
theory: it is enough to have a mathematical description of the set of
models.


\section{No-go theorems based on polynomial-time classical computation and possible post-selection}

\label{sect:no:go}

In this section I derive a no-go theorem that rules out some theories
based on sequential classical computation, even in the presence of
post-selection. The theorem is based on the result of
Sect.~\ref{sect:definition} and relies on widely accepted conjectures.
The idea is to prove that a problem is \emph{not} hard for
class {\bf SampBQP} by proving that it belongs to a class {\bf X},
such that:
\begin{itemize}
\item {\bf X} does not contain {\bf SampBQP} (i.e. there is some problem of
  {\bf SampBQP} that is not inside {\bf Y}), and
\item {\bf X} is closed under polynomial-time reductions (i.e. any
  problem that can be polynomial-time reduced to a problem in {\bf X}
  is in {\bf X}).
\end{itemize}
In turn, the lack of hardness for {\bf SampBQP} prevents the theory to
be valid.

\subsection{Class expressing the idea of sequential calculation}

In sequential models, the hidden variables take a value at $t=0$ and
evolve, step by step, to the final time, at which they are evaluated.
For models of circuits, a reasonable assumption, similar to locality,
is that the calculation of the next step involves groups of variable
of fixed size. Intuitively, the models of such a theory $T$ can be
sampled in polynomial time in the size of the model representation.
This kind of theories is defined as ``sequential'' and they are
causal.

To formalize this idea, I provide a computational complexity class in
which SAMP$(T)$ of a sequential theory should fall, {\bf SampP}.  I
define it following Ref.~\cite{aaronson2010}.

\begin{definition}[{\bf SampP} sampling class]
  The class {\bf SampP} is the class of sampling problems
  $S=\left\{ \mathcal{D}_x(y) \, | \, x\in \{0,1\}^* \right\}$ for which
  there exists a probabilistic polynomial-time algorithm $B$ that:
  \begin{enumerate}
    \item Takes an input $x$,$0^{1/\epsilon}$ (notice that $0^k$ means a
      string of $k$ 0s);
    \item Returns a string $y\in \{0,1\}^{p\left(|x|\right)}$;
    \item The probability of getting a $y$ on a given input $x$,
      $\mathcal{C}_x(y)$, is close to $\mathcal{D}_x(y)$ as:
      \begin{equation}
        \sum_{y\in \{0,1\}^{p\left(|x|\right)}}
        \left| \mathcal{C}_x(y) - \mathcal{D}_x(y) \right| \le \epsilon
      \end{equation}
  \end{enumerate}
\end{definition}

This class is the classical analogous of {\bf SampBQP},
Def.~\ref{def:samp:bqp}. It captures the idea of generating samples,
according to a prescribed model, with an algorithm operating in
polynomial time in the length of its input, matching a given
probability distribution within a given error $\epsilon$.

According to a widely accepted conjecture on quantum computational
complexity classes, this class does not include {\bf SampBQP}.

\begin{proposition}
  \label{prop:sampp:not:circuit}
  {\bf SampBQP} $\nsubseteq$ {\bf SampP}. This proposition is based on
  the conjecture that {\bf BQP} is not included in {\bf BPP}.
\end{proposition}

The proof is given in Appendix~\ref{sect:no:go:sampling:post:selection}.

\subsection{Class expressing the idea of post-selection}

The idea of post-selection is that a probabilistic algorithm first
generates a sample $y$, then it decides whether to accept or reject
it. As an example, Monte Carlo rejection sampling algorithms work in
this way. In the context of hidden-variable models, we can also see an
analogy with the retro-causality induced by the measurement
setting. For example, in Schulman model we assign starting values to
the $\varphi_j$, we make them sequentially evolve, through
deterministic steps and probabilistic kicks, but, in the end, the
measurement imposes a selection.

The idea of an algorithm that generates samples by sequential
calculation followed by post-selection is expressed by the following
definition.

\begin{definition}[{\bf PostSampP} sampling class]
  \label{def:post:samp:p:lenient}
  The class {\bf PostSampP} is the class of sampling problems
  $S=\left\{ \mathcal{D}_x(y) \, | \, x\in \{0,1\}^* \right\}$ for which
  there exists a probabilistic polynomial-time algorithm $B$ that:
  \begin{enumerate}
    \item Takes an input $x$,$0^{1/\epsilon}$ (notice that $0^k$ means a
      string of $k$ 0s);
    \item Either returns a string $y\in \{0,1\}^{p\left(|x|\right)}$
      or 'FAILED';
    \item The probability of getting a valid string (i.e. not
      returning 'FAILED') does not vanish;
    \item Called $\mathcal{C}_x(y)$ the probability of getting a $y$ on a given input $x$, conditioned to get a valid string (i.e. not
      returning 'FAILED'), is close to $\mathcal{D}_x(y)$ as:
      \begin{equation}
        \sum_{y\in \{0,1\}^{p\left(|x|\right)}}
          \left| \mathcal{C}_x(y) - \mathcal{D}_x(y) \right| \le \epsilon
      \end{equation}
  \end{enumerate}
\end{definition}

It is easy to see that a post-selection procedure fits in
Def.~\ref{def:post:samp:p:lenient}: in case of acceptance, the
algorithm outputs the generated sample $y$, else, it outputs 'FAILED'.

The reverse also holds. Consider any algorithm $M$ that matches the
Def.~\ref{def:post:samp:p:lenient}: it either returns a sample $y$ or
'FAILED'. We can tweak it adding a fake post-selection, simply
returning $y$ when $M$ returns a valid sample $y$ and reporting the
failure when $M$ output 'FAILED' (The concept is so simple that it is
even hard to explain).

We thus conclude that the algorithms $M$ of
Def.~\ref{def:post:samp:p:lenient} actually describe post-selection,
even if the operation of post-selecting is not explicitly enforced in
the definition.

Although a post-selection algorithm could appear ``lazy'', being
allowed to refuse to give an output, the overall power of
post-selection algorithms is more than equivalent algorithms without
post-selection. To understand this point, it must be noticed that each
run of the algorithm takes polynomial time in $|x|$, but getting a
valid sample $y$, i.e. not 'FAILED', requires to iterate the algorithm
several times. Possibly, this leads to very long calculation time,
much longer than polynomial. For this reason, the post- versions of
the classes are expected to be harder to calculate than the non-post-
versions.

It is possible to prove that this class does not include {\bf
  SampBQP}, based on a conjecture on quantum computational complexity
classes. This conjecture is quite strong but it is supported by a few
clues~\cite{aaronson2010bis}.

\begin{proposition}
  \label{prop:postsampp:not:circuit}
  {\bf SampBQP} $\nsubseteq$ {\bf PostSampP}. This proposition is
  based on the conjecture that {\bf BQP} is not included in {\bf PH}.
\end{proposition}

The proof is given in
Appendix~\ref{sect:no:go:sampling:post:selection}.

\subsection{A more severe class expressing the idea of post-selection}

The two sampling classes {\bf SampP} and {\bf PostSampP} specify the
error as additive (as opposed to a multiplicative error, representing
a fraction of the value). Due to this condition on the error, the
class {\bf PostSampP} can be considered ``lenient'' with respect the
following, ``severe'' class, having a stricter condition on the
error. The reason for introducing this severe version is that the
proofs related to this severe version are based on a weaker
conjecture, only based on classical computation classes.

\begin{definition}[{\bf PostSampP*} sampling class, severe]
  \label{def:post:samp:p:severe}
  The class {\bf PostSampP*} is defined in the same way of {\bf
    PostSampP}, except for the error definition, which is substituted
  by the following inequality.
  \begin{equation}
    \left| \frac{\mathcal{C}_x(y)}{\mathcal{D}_x(y)} -1 \right| \le \epsilon
  \end{equation}
  for every $y\in \{0,1\}^{p\left(|x|\right)}$. Notice that
  $\mathcal{C}_x(y)$ is already conditioned to have a valid sample,
  which means that the output is not 'FAILED'.
\end{definition}

According to a widely accepted conjecture, this class does not include
{\bf SampBQP}. At variance with the proof of
Prop.~\ref{prop:postsampp:not:circuit}, which refers to {\bf
  PostSampP}, the requrested conjecture only relies on classical
computational complexity classes; moreover, the conjecture is widely
accepted.

\begin{proposition}
  \label{prop:postsampp:star:not:circuit}
  {\bf SampBQP} $\nsubseteq$ {\bf PostSampP*}.  This proposition
  is based on the conjecture that the polynomial hierarchy does not
  collapse.
\end{proposition}

The proof is given in
Appendix~\ref{sect:no:go:sampling:post:selection}.

Since the condition on the error is more severe, {\bf PostSampP*}
$\subseteq$ {\bf PostSampP}.

\subsection{The no-go theorem}

\begin{figure}
  \begin{center}
    \includegraphics{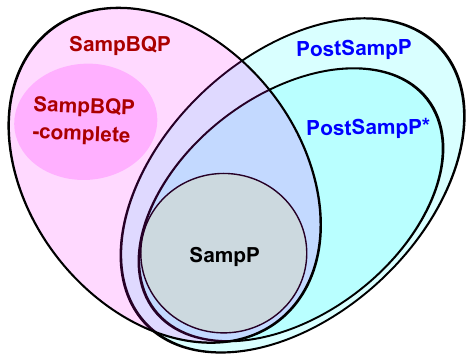}
  \end{center}
  \caption{Comparison of the classes of sampling problems {\bf SampP},
    {\bf PostSampP}, {\bf PostSampP}, and {\bf SampBQP}.}
  \label{fig:hierarchy}
\end{figure}

The relation between the classes is graphically shown in
Fig.~\ref{fig:hierarchy}. We see that {\bf PostSampP} and {\bf
  PostSampP} actually extend {\bf SampP}, but in a different direction
with respect to {\bf SampBQP}, thus missing {\bf SampBQP-complete}.

We can now summarize the results obtained so far with the following theorem.

\begin{theorem}
  \label{theorem:no:go}
  A theory $T$ is given.
  \begin{itemize}
  \item If SAMP$(T) \in$ {\bf SampP}, then
    $T$ is not valid. This proposition is based on the conjecture that
    {\bf BQP} is not contained in {\bf BPP}.
  \item If SAMP$(T) \in$ {\bf PostSampP}, then
    $T$ is not valid. This proposition is based on the conjecture that
    {\bf BQP} is not contained in {\bf PH}.
  \item If SAMP$(T) \in$ {\bf PostSampP*}, then
    $T$ is not valid. This proposition is based on the conjecture that
    the polynomial hyerarchy does not collapse.
  \end{itemize}
\end{theorem}

\begin{proof}
  The proof relies on Props.~\ref{prop:sampp:not:circuit},
  \ref{prop:postsampp:not:circuit}, and
  \ref{prop:postsampp:star:not:circuit}, stating that {\bf SampP},
      {\bf PostSampP}, and {\bf PostSampP*} do not contain {\bf
        SampBQP}, based on the mentioned conjectures.  Moreover, it is
      trivial to see that {\bf SampP}, {\bf PostSampP}, and {\bf
        PostSampP*} are closed under polynomial-time reductions.
\end{proof}

The meaning is that a hidden-variable theory $T$, whose associated
problem SAMP$(T)$ is in {\bf SampP}, {\bf PostSampP}, or {\bf
  PostSampP*}, is not valid, because it is too computationally
simple.

The class {\bf SampP} captures the idea of sampling from sequential
and local theories. Actually, such theories would describe systems
similar to classical, sequential computer calculations: quantum
mechanics is harder than this.  In a very vivid way, this concept has
been expressed in a paper with the title ``the universe is not a
computer''~\cite{wharton2015}. This result is however already expected
from Bell's theorem.

The classes {\bf PostSampP} and {\bf PostSampP*} represent theories
with retro-causality implemented by post-selection. Theories with
associated sampling problem in these classes are retro-causal, thus
they are not ruled out by Bell's theorem. They are instead ruled out
by the no-go theorem presented here: their associated sampling problem
is harder than for sequential theories, but not enough to reach the
computational complexity of quantum circuits.

It is important to notice that the no-go theorem has quite strong
hypothesis, which can violated by theories, as in every no-go
theorem. In particular, retro-causal theories are not ruled out as a
whole; rather, the no-go theorem says that retro-causality alone is
not enough to make a theory valid. Indeed, in
Sect.~\ref{sect:no:go:post:selection:schulman}, I actually show that
the generalized Schulman theory escapes the hypothesis of the no-go
theorem.


\section{Application of the no-go theorem to the generalized Schulman theory and some of its variants}

\label{sect:no:go:post:selection:schulman}

In this section I discuss the theory obtained by generalizing the
Schulman model, $T_{eS}$, already introduced in
Sect.~\ref{sect:schulman:model}. I provide two simplified theorems and I
analyze them by means of the no-go theorem. The comparison among
the three models highlights the importance of various details of the
extended Schulman theory.

\subsection{The variants of the Schulman model}

Before introducing the two variants, it is useful to notice that the
extended Schulman theory, introduced in Sect.~\ref{sect:schulman:model},
requires to make two limits. The first is that the width of the
Lorentz distribution vanishes: $\delta \varphi_L\to 0$. The second is
more implicit; let us assume that the constraint imposed by the
measurement setting has a tolerance $\delta \varphi_M$:
\begin{equation}
  \left| \varphi_j - \vartheta_j - n\pi \right| < \delta \varphi_M
\end{equation}
or
\begin{equation}
  \left| \varphi_j - \vartheta_j - \pi/2 - n\pi \right| < \delta \varphi_M
\end{equation}
The model prescribes that also $\delta \varphi_M$ vanishes, $\delta
\varphi_M\to 0$. The meaning and importance of these limits will be
discussed below.

In the first variant of the theory, $T_{nl}$, I remove these limits
(nl = no limits).  This means that I assume that $\delta \varphi_L$
and $\delta \varphi_M$ have a fixed, small value, e.g. 1/1000 (notice
that the unit is rad). Moreover, in $T_{nl}$, the probability follows
a truncated Lorentz distribution:
\begin{equation}
  \label{eq:truncated:lorentz:distribution}
  \mathcal{P}\left(\Delta \varphi\right) =
  \left\{
  \begin{array}{lc}
    \frac
        {\mathcal{M} \delta \varphi_L}
        {\delta \varphi_L^2 + \Delta \varphi^2} & \Delta \varphi \le
        \frac{1}{\delta \alpha} \\
        0 & \Delta \varphi >
        \frac{1}{\delta \alpha}  
  \end{array}
  \right.
\end{equation}
where $\delta \alpha$ is a constant defining the width of the
truncation and $\mathcal{M}$ is a normalization constant, approaching
1 for $\delta \alpha\to 0$. Clearly, this distribution approaches the
Lorentz distribution, Eq.~\ref{eq:lorentz}, for $\delta \alpha\to 0$.

In the second variant, $T_{c}$ (causal), I further simplify $T_{nl}$,
by removing the retro-causality. At the instant of the measurement,
the values of $\varphi_j$ are simply measured, without imposing any
constraint.

\subsection{Sequential and approximated retro-causal theories}

\begin{figure}
  \includegraphics{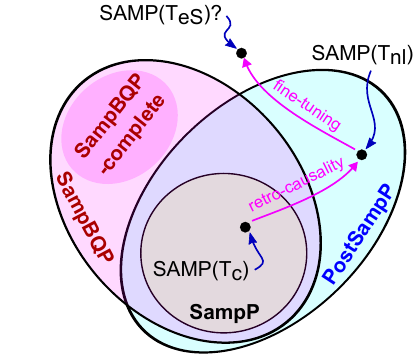}
  \caption{Scheme of the complexity classes and the position of the
    sampling problems of the three variants of the generalized
    Schulman theory: $T_{c}$, the sequential theory; $T_{nl}$, theory
    based on sequential calculation with post-selection; $T_{eS}$, the
    full generalized Schulman theory. }
  \label{fig:state:prob}
\end{figure}

It is easy to devise an algorithm that samples from $T_{c}$, trivially
following the described procedure: generate the initial values of the
variables, apply the prescribed random kicks and gates, and measure
the resulting $\varphi_j$. I remind that the characteristic feature of
$T_{c}$ is the absence of constraints imposed by the measurement,
thus, there is no retro-causality.

Intuitively, these operations take polynomial time in the size of the
representation of the model. Under reasonable hypothesis on the
computational complexities of the layers, it can be shown that the
operation can be actually done in polynomial time and with the
prescribed precision, thus SAMP$(T_{c})$ $\in$ {\bf SampP}
(See Fig.~\ref{fig:state:prob}).

This fact tells us that $T_{c}$ is not valid.  This result was already
expected from the Bell's theorem: having removed the retro-causality,
the model has no chances to be valid. But here we recover it through a
different way: $T_{c}$ is too simple to calculate to model such a
complex system as quantum mechanics.

Adding the retro-causality, we go from $T_{c}$ to $T_{nl}$ (we will
arrive to the full $T_{eS}$, but we make this step by step). The
retro-causality is added in the form of a post-selection: the
resulting samples must match with the measurement settings. This
suggests that SAMP$(T_{nl})$ can be more computationally complex than
SAMP$(T_c)$ and, likely, it falls in {\bf PostSampP}. Actually this
can be proved under the same reasonable assumptions (See
Fig.~\ref{fig:state:prob}).

This increase in computational complexity is still not enough to make
SAMP$(T_{nl})$ {\bf SampBQP}-hard, because {\bf SampBQP-complete} is
even outside {\bf PostSampP}, so $T_{nl}$ is still not a valid
theory. Even if $T_{nl}$ is harder than $T_{nl}$, it is still not hard
enough to be a valid theory.

\subsection{The full generalized Schulman theory}

Going from $T_{nl}$ to the full, general Schulman model $T_{eS}$
requires to take a limit: instead of having finite $\delta \varphi_L$,
$\delta \varphi_M$, and $\delta \alpha$, the limits for vanishing
values must be taken. Under these limits, the situation changes and
the reasoning that we applied to $T_{nl}$ cannot be carried out for
$T_{eS}$. The main difference is that it is impossible to sample from the
Lorentz distribution Eq.~\ref{eq:lorentz} with vanishing width $\delta
\varphi_L\to 0$, as required for $T_{eS}$, while it is possible to
sample from Eq.~\ref{eq:truncated:lorentz:distribution}, with finite
$\delta \varphi_L$ and $\delta \alpha$, as required by $T_{nl}$. A second
difference is that, for vanishing measurement tolerance $\delta
\varphi_M\to 0$, as requrested by $T_{eS}$, all the generated samples
are rejected.

These reasons prevent to prove that SAMP$(T_{es})$ is inside {\bf
  PostSampP}. I give a clue suggesting that SAMP$(T_{es})$ could be
harder than {\bf PostSampP}. The mathematical procedure that is used
in literature to solve the extended Schulman models is symbolic; it is
based on the Euler's solution of the Basel problem (see
Ref.~\cite{wharton2020}, appendix). As a first step, the tolerance of
the measurement setting is brought to 0, $\delta \varphi_M\to
0$. Then, it is noticed that the possible configurations of the kicks
are the ones in which the maximum number of the $\Delta \varphi$
vanish. The procedure then considers all the possible configurations
with the same, maximal, number of vanishing kicks $\Delta
\varphi$. The number of these configurations scales as the exponential
of the total number of kicks. Thus it seems that the symbolic
procedure resorts on something similar to counting, which is likely
harder than {\bf PostSampP}.

At the current state, it seems that it is not possible to apply the
no-go theorem to the full generalized Schulman theory $T_{eS}$ to
conclude that it is not valid. This fact does not prove that is is a
viable model: it is possible that there is an alternative proof that
SAMP$(T_{es})$ is inside {\bf PostSampP}, or that, even if
SAMP$(T_{es})$ is outside {\bf PostSampP}, it is not {\bf
  SampBQP}-hard, or that it is {\bf SampBQP}-hard but there is no
meaningful mapping between quantum circuits and the models. However,
it encourages further efforts in this direction.

The position of the SAMP$(T)$ problems with respect to the
computational complexity classes is schematically shown in
Fig.~\ref{fig:state:prob}.  SAMP$(T_{c})$ is in {\bf SampP},
SAMP$(T_{nl})$ is in {\bf PostSampP} and, likely, outside {\bf
  SampP}. We do not have evidences of inclusion SAMP$(T_{eS})$ in
these two classes, thus I represent it outside them. If so, it is
possible that SAMP$(T_{eS})$ is {\bf SampBQP}-hard, which means that
$T_{eS}$ could be valid; this makes it the only candidate, among the
variants considered here, as a viable model.

We can notice that $T_{nl}$ approximates $T_{eS}$ for vanishing
$\delta \varphi_L\to 0$, $\delta \varphi_M\to 0$, and $\delta
\alpha$. However, as discussed above, this approximation cannot be
used to bring the results of the post-selection algorithm close to the
desired result of $T_{eS}$; the error is, in any case, excessive with
respect to the requirements of {\bf PostSampP}. This fact has a strong
analogy with the requirement of fine tuning~\cite{wood2015}: even a
small deviation from the ideal parameters prevent the model to be
valid.


\section{Conclusion}

\label{sect:conclusion}

I provided a no-go theorem for hidden-variable theories, based on
computational complexity theory. The theory $T$ is first associated to
the mathematical problem of sampling from it, i.e. to generate samples
distributed with the probability defined by the model, SAMP$(T$.  Then
the computational complexity of the sampling problem is compared to
three class: {\bf SampP}, {\bf PostSampP}, and {\bf PostSampP*}. If it
is in one of them, i.e. it is less hard than, or at most as hard as
that class of problems, then it is not hard enough to model quantum
mechanics and the conclusion is that $T$ is not a valid theory.

It is important to notice that, like the Bell's theorem, this no-go
theorem is a necessary condition imposed on the possible
implementations of the model, without the need of associating a
specific implementation of the model to a quantum circuit.

Informally, the idea is to try to write a computer program that
samples from the hidden-variable model, within a given error, in a
time (number of steps) that is polynomial in the size of the
circuit. If we give to the algorithm the possibility of making
post-selection (rejecting a generated sample), we get the class {\bf
  PostSampP} or {\bf PostSampP*}. Without post-selection, we get the
smaller {\bf SampP}. Hidden-variable theories for which we can generate
samples in this way are too computationally simple to represent
quantum mechanics. See Fig.~\ref{fig:hierarchy} for a graphical
representation of the relations between the classes and the sampling
problems.

The formalization of this idea is the no-go theorem. Rigorously, it
relies on conjectures, which are however quite well accepted by the
scientific community. The use of {\bf SampP} relies on the widely
accepted conjecture that {\bf BQP} is not contained in {\bf BPP},
roughly, that quantum computers are more powerful than classical
computers. The use of {\bf PostSampP} is based on the stronger
conjecture that {\bf BQP} is not even contained in the large class
{\bf PH}. Although this fact is conjectured to be true, the consensus
is much less strong. For this reason, I provide an alternative class,
{\bf PostSampP}; for this class, the used conjecture is that the
polynomial hierarchy does not collapse. This is a widely accepted
conjecture, representing a generalization of the widely known
conjecture that {\bf P} does not equal {\bf NP}, and most of the
computational complexity theory relies on it.

Intuitively, we can associate {\bf SampP} to theories that represent
the sequential evolution of a state with local laws; in addition, {\bf
  PostSampP} and {\bf PostSampP*} include post-selection. Such
theories are deemed not valid by the no-go theorem. The first class
was already ruled out by the Bell's theorem and is further ruled out
by this independent no-go theorem. The second class is instead only
identified by this approach.

Like for every no-go theorem, the interest resides mostly in the the
hypothesis which can be violated (see the example of the Bell's
theorem). Among the possible violations, there are the following.
\begin{itemize}
\item The hidden-variable theories ruled out by the no-go theorem are
  based on sequential calculation with a possible post-selection.
  This is only one of the possible schemes of causality violation and
  retro-causality. All-at-once theories, not based on post-selection,
  could escape the hypothesis of the no-go theorem.
\item The post-selection required by the hypothesis of the no-go
  theorem requires that the validity of the generated sample is
  checked in polynomial time. A way circumvent the theorem is to
  develop a theory in which it is harder (e.g. exponential time) to
  decide to accept or reject a sample.
\item The classses {\bf SampP}, {\bf PostSampP}, and {\bf PostSampP}
  are based on a sampling within a limited error on the resulting
  distribution. Similar models can lead to widely different errors
  under limited resources. Actually, this situation reminds us the
  requirement for the fine tuning~\cite{wood2015}.
\item The verdict of non-validity of the no-go propositions can be
  circumvented by specifying a mapping $R$ between quantum circuits
  and models whose calculation requires more time than polynomial in
  the size of the input. This is a possibility, although it is hard to
  believe that such a hidden-variable model would be
  credible. Usually, the association between a quantum circuit and its
  corresponding hidden-variable circuit is thought to be direct, not
  requiring excessive calculations.
\end{itemize}

Moreover, I remind once again that the proof that {\bf SamopBQP} is
not included in {\bf SampP}, {\bf PostSampP}, nor {\bf PostSampP} is
based on conjectures. They are credible and widely held in the
computational complexity theory. This notwithstanding, it must be
considered that there is the possibility that these conjectures will
be proven false in the future.

A curious remark is that, according to the analysis above, here we aim
at models that are more difficult to calculate than others: only such
models can be valid models for quantum mechanics. The general trend is
instead the opposite, i.e. looking for fast algorithms to spare time.

As an example of the application of the no-go theorem, I formalize a
generalization of the Schulman model, $T_{eS}$, a retro-causal model
known to simulate some of the quanum phenomena. The generalization
encompasses various schemes that have been proposed in
literature. Then, I give two simplifications. In the first, $T_{nl}$,
some parameters which should vanish are approximated with finite
values; it violates the fine tuning. In the second, $T_c$, a furhter
change is made, consisting in removing the retro-causality.

The computational complexity analysis shows that $T_c$ belongs to {\bf
  SampP}, thus $T_c$ is not valid. This result is in agreement with
Bell's theorem, since $T_c$ is obtained by removing the
retro-causality. See Fig.~\ref{fig:state:prob} for a graphical
representation.

By restoring the retro-causality in $T_c$, we first obtain
$T_{nl}$. The added retro-causality brings the problem into a harder
class, {\bf PostSampP}: post-selection is analogous to the constraint
imposed by the measurement settings. However, the increase in
computational complexity of {\bf PostSampP} with respect to {\bf
  SampP} is not enough, or better, is not in the right direction to
make $T_{nl}$ valid. This is expected, since $T_{nl}$ is obtained from
$T_{eS}$ by removing the fine tuning, a necessary
feature~\cite{wood2015}.

It cannot be proved that the sampling problem associated to the full
generalized Schulman model, $T_{eS}$, is not in {\bf PostSampP}. This
prevents us to apply the no-go theorem. This hint encourages furhter
work on the model. See Fig.~\ref{fig:state:prob} for a graphical
representation.

We thus see that the no-go theorem allows us to obtain results that
are similar to known ones, but from a completely different
perspective.  The analysis of the variants of the generalized Schulman
theory highlights two of the features that are important:
retro-causality and fine tuning. Roughly, the former brings the
computational complexity from {\bf SampP} to {\bf PostSampP},
requiring post-selection; the latter violates the conditions on the
error by requiring a fine tuning, in the form of a limit for vanishing
parameters.

Similar features can be further exploited in future models. In
particular, a possible line of research could be to better understand
the computational complexity of SAMP$(T_{es})$. It can be
calculated as the limit for vanishing parameters (fine tuning) of
SAMP$(T_{nl})$, however, this does not tell us much about its
computational complexity class. Indeed, the mathematical evaluation of
the model does not rely on this limit, but rather on an exact
calculation based on the Euler's solution of the Basel problem (see
Ref.~\cite{wharton2020}, appendix). Developing an algorithm based on
this idea could be beneficial not only for the theoretical analysis of
the computational complexity, but also for the practical aim of
calculating the outcome of these hidden-variable models.

Is is worth further discussing the key feature of the algorithms used
in {\bf PostSampP}, the post-selection.  It actually increases the
computational complexity with respect to {\bf SampP}, but it increases
it in the wrong way, i.e. not in the direction of quantum mechanics
(compare {\bf PostSampP} and {\bf SampBQP} in
Fig.~\ref{fig:hierarchy}), so {\bf PostSampP} fails to include {\bf
  SampBQP}. Computational complexity theory tells us what is this
additional, but useless, complexity: post-selection alogorithms are
able to solve puzzles, i.e. problems in which the solution can be
easily checked, as opposed to games, like chess, in which there is no
simple way to say if a move is good or not, not even when we see
it. Technically, for example, {\bf PostBPP} and {\bf BPP}
$^{\mathrm{PostSampP}}$ include {\bf NP} (roughly, the
puzzles). Solving puzzles does not get a speed up from quantum
computers, thus {\bf PostSampP} is not included in {\bf SampBQP}, nor
solving puzzles suffices to simulate quantum computers, e.g. {\bf
  SampBQP} is not included in {\bf PostSampP}. These two abilities,
solving puzzles and solving quantum systems, seem to be
unrelated. Thus, the retro-causality alone does not suffices to
simulate quantum circuits.

It is worth noticing that this additional complexity also encompasses
the ability of solving all-at-once models in which a solution can be
checked in polynomial time. We know that such models are plagued by
the problem of propagation of signals back in time, which, in turns,
results in the presence of paradoxes. More technically, as we already
saw, {\bf BPP}$^{\mathrm{PostSampP}}$ include {\bf NP}, and {\bf NP}
is a class which can deal with things like $x=NOT x$: paradoxes. As we
already saw, the violation of the fine tuning, operated in $T_{nl}$,
is expected to induce the propagation of signals back in time; indeed,
we saw that SAMP$(T_{nl})$ belongs to {\bf PostSampP}, and thus
has this ``useless'' additional ability, to solve puzzles and deal
with paradoxes.

According to this discussion, it is thus important to take into
considerations and try to develop hidden-variable models in which the
verification of a solution (or the evaluation of its probability)
cannot be done in polynomial time. This would go beyond the abilities
of post-selection algorithms, thus violating the hypothesis of the
no-go theorem.

\appendix


\section{Classes of decision problems related to post-selection}

\label{sect:decision}

This section contains a discussion of known facts in computational
complexity theory and the derivation of relations between complexity
classes related to post-selection. They are needed for the proofs
given in Sect.~\ref{sect:no:go:sampling:post:selection}.

\subsection{Definition of the bounded-error classes}

There are two classes representing the decision problems that can be
solved with bounded error, in polynomial time, by probabilistic
algorithms: {\bf BPP} and {\bf BQP}. They are based on classical and
quantum algorithms, respectively.

I report the definition of {\bf BPP}, even if it does not appear in
the discussion below, for the sake of comparison.

\begin{definition}[Class {\bf BPP}]
  \label{def:bpp}
``The language $L$ is in {\bf BPP}'' means that there esists a
  probabilistic (classical) algorithm $M$, such that:
\begin{enumerate}
\item $M$ takes an input $x\in \{0,1\}^*$ and outputs 1 bit, $y$;
\item For all $x\in L$, the probability of getting output $y=1$ is
  greater than or equal to 2/3
\item For all $x\notin L$, the probability of getting output $y=0$ is
  greater than or equal to 2/3
\end{enumerate}
\end{definition}

For the class {\bf BQP}, I use the definition based on quantum circuits.

\begin{definition}[Class {\bf BQP}]
  \label{def:bqp}
``The language $L$ is in {\bf BQP}'' means that there 
 exists a polynomial-time uniform family of quantum circuits 
$\left\{Q_n : n\in \mathbb{N} \right\}$, such that:
\begin{enumerate}
\item For all $n\in \mathbb{N}$, $Q_n$ takes $n$ qubits as input and
  outputs 1 bit, $y$;
\item For all $x\in L$, the probability of getting output $y=1$ is
  greater than or equal to 2/3
\item For all $x\notin L$, the probability of getting output $y=0$ is
  greater than or equal to 2/3
\end{enumerate}
\end{definition}

I remind that a ``polynomial-time uniform family of quantum circuits''
is a family of descriptions of quantum circuits that are generated in
polynomial time in $n$ by a Turing machine. This definition also
implies that the size of the description of the circut $Q_n$ is also
polynomially bounded.

\subsection{Definition of the post-selection versions}

The discussion below will involve the variants of these bounded-error
classes, in which post-selection is added. The class {\bf PostBPP} is
defined in Ref.~\cite{bravyi2007}.

\begin{definition}[Class {\bf PostBPP}]
  \label{def:postbpp}
``The language $L$ is in {\bf PostBPP}'' means that there is a
  probabilistic algorithm $M$ such that:
\begin{enumerate}
\item $M$ runs for polynomial time on the size of its input;
\item $M$ returns two bits, $y_{\mathrm{sample}}$ and $y_{\mathrm{valid}}$;
\item The probability of getting $y_{\mathrm{valid}}=1$ does not vanish;
\item For all $x\in L$, conditioned to $y_{\mathrm{valid}}=1$, the probability
  of getting output $y_{\mathrm{sample}}=1$ is greater than or equal to 2/3
\item For all $x\notin L$, conditioned to $y_{\mathrm{valid}}=1$, the
  probability of getting output $y_{\mathrm{sample}}=0$ is greater than or equal
  to 2/3
\end{enumerate}
\end{definition}

The class {\bf PostBQP} is defined in Ref.~\cite{aaronson2005}.

\begin{definition}[Class {\bf PostBQP}]
  \label{def:postbqp}
``The language $L$ is in {\bf BQP}'' means that there 
 exists a polynomial-time uniform family of quantum circuits 
$\left\{Q_n : n\in \mathbb{N} \right\}$, such that:
\begin{enumerate}
\item For all $n\in \mathbb{N}$, $Q_n$ takes $n$ qubits as input and
  outputs 2 bits, $y_{\mathrm{sample}}$ and $y_{\mathrm{valid}}$;
\item The probability of getting $y_{\mathrm{valid}}=1$ does not vanish;
\item For all $x\in L$, conditioned to $y_{\mathrm{valid}}=1$, the probability
  of getting output $y_{\mathrm{sample}}=1$ is greater than or equal to 2/3
\item For all $x\notin L$, conditioned to $y_{\mathrm{valid}}=1$, the
  probability of getting output $y_{\mathrm{sample}}=0$ is greater than or equal
  to 2/3
\end{enumerate}
\end{definition}

We can interpret $y_{\mathrm{valid}}$ as a flag saying whether $y_{\mathrm{sample}}$
is valid or has to be rejected. In case of rejection, it is possible
to run again the algorithm, until a valid $y_{\mathrm{sample}}$ is returned. By
definition, a valid sample will be obtained, soon or later, even if it is not
guaranteed to happen within any known time.

To better clarify, the conditioned proabilities refer to the case in
which a valid value $y_{\mathrm{sample}}$ is returned, i.e. neglecting the
$y_{\mathrm{valid}}=0$ outcomes. For example, given the algorithm $M$
and the input $x$, if the overall prpbability of getting $y_{\mathrm{sample}}=0$,
$y_{\mathrm{valid}}=1$ is 2\%, of getting $y_{\mathrm{sample}}=1$,
$y_{\mathrm{valid}}=1$ is 6\%, and thus the overall probability of
getting $y_{\mathrm{valid}}=0$ is 92\%, then the conditioned
probability of getting $y_{\mathrm{sample}}=1$ is 3/4 = 6\%/(2\%+6\%).

It is clear that the term Post- here has the same meaning explained
for the sampling classes, i.e. it refers to a post-selection
algorithm. The non-validity expressed by $y_{\mathrm{valid}}=0$
informally corresponds to the 'FAILED' output of the post-selection
sampling algorityhms.

\subsection{The error is arbitrary}

In all the four defined bounded-error classes, the parameter 2/3
represents a success rate of the probabilistic algorithm. In other
terms, $\epsilon=$ 1/3 is the maximum admissible error (hence the name
``bounded-error'' which appears in the class acronyms {\bf BPP} and
{\bf BQP} and in their post-selection versions).

It is worth noting that the error $\epsilon$ is actually completely
arbitrary: defining a class with a different $\epsilon<1/2$ leads to
the very same class. The idea is that, for $\epsilon<1/2$, any
precision can be obtained by repeating the algorithm a suitable number
$m$ of times, independent of the input $x$, and taking the majority of
the results.

The operation is trivially done for {\bf BPP}. In the case of {\bf
  BQP}, this operation can be simply done by considering a quantum
circuit formed by repeating $m$ times the original one and finally
taking the majority of the outputs, operation that can be done in
polynomial time by quantum circuits. This factor $m$ is fixed by the
desired error, thus it does not change the condition that the family
of quantum circuits is polynomially generated.

In the case of {\bf PostBPP}, something similar can be done.  The
following explanation refers to the explicit post-selection.  Let us
assume that we need $m$ samples to decide the result within the
required error $\epsilon$. Then, the algorithm generates $m$ samples
$y_j$, then operates a post-selection decision on all of them. If any
of the samples is rejected, then $y_{\mathrm{valid}}=0$ is returned,
together with any $y_{\mathrm{sample}}$ (0, 1, or random, it does not
matter). Else, the outputs from each of the $y_{\mathrm{sample},j}$ is
calculated and the majority of them defines the output
$y_{\mathrm{sample}}$, together with $y_{\mathrm{valid}}=0$.

The following explanation refers to the algorithm $M$ described in
Def.~\ref{def:postbpp}.  The algorithm generates $m$ samples $y_{\mathrm{sample},j}$
and $y_{\mathrm{valid},j}$.  If $y_{\mathrm{valid},j}=0$ for any of
the samples $j$, then $y_{\mathrm{valid}}=0$ is returned, together
with any $y_{\mathrm{sample}}$ (0, 1, or random, it does not matter). Else, the
majority of the $y_{0,j}$ is used to calculate the output $y_{\mathrm{sample}}$,
together with $y_{\mathrm{valid}}=1$. An analogous procedure applies
to {\bf PostBQP}.

\subsection{Relation between {\bf PostBQP} and {\bf PostBPP}}

The following sections rely on two relations between classes. The
first is expressed by this proposition.

\begin{proposition}
  \label{prop:postbpp:postbqp}
        {\bf PostBPP} does not include {\bf PostBQP}.  This
        proposition is based on the conjecture that the polynomial
        hierarchy does not collapse.
\end{proposition}

This means that there is at least one language $L$ that belongs to
{\bf PostBQP} but not to {\bf PostBPP}.

In computational complexity theory, most of the propositions and
theorems that are proven state the inclusion of classes. Proving that
a class is not included in another one is much less common, only a few
of such statements are known. However, such proofs can be found
assuming the validity of some conjectures, which are widely believed
to be true. Also in the case of Prop.~\ref{prop:postbpp:postbqp}, I
will give a proof based on a conjecture. The required conjecture is
that the polynomial hierarchy does not collapse. It is a
generalization of the conjecture that {\bf NP} does not equal {\bf
  P}. This conjecture is widely assumed to hold and many results in
computational complexity theory rely on it.

\begin{proof}
  This proof is based on the conjecture that the polynomial hierarchy
  does not collapse.

  Ref.~\cite{aaronson2005} shows that {\bf PostBQP} equals {\bf PP}.
  Under the conjecture that the polynomial hierarchy does not
  collapse, {\bf PP} is not included in {\bf PH}
  (Ref.~\cite{toda1991}), thus {\bf PostBQP} is not included in {\bf
    PH} as well. Thus, there is a language $L_0$ such that $L_0\in$
  {\bf PostBQP} and $L_0\notin$ {\bf PH}.

  Ref.~\cite{bravyi2007} shows that {\bf PostBPP} equals {\bf
    BPP$_{path}$}, defined in Ref.~\cite{han1997}. This class is
  included in {\bf BPP$^{\mathrm{NP}}$} (Ref.~\cite{han1997}), which,
  in turn, is included in {\bf $\Sigma^p_3$} (see Ref.~\cite{han1997},
  Fig. 1; the inclusion is considered trivial) and is thus part of
  {\bf PH} (see the definition of {\bf PH}).

  Summarizing, {\bf PostBPP} $\subseteq$ {\bf PH}. Since $L_0\notin$ {\bf
    PH}, then $L_0\notin$ {\bf PostBPP}.
\end{proof}

\subsection{Relation between {\bf BQP} and {\bf PostBPP}}

This proposition represents a relation that is more strict than
Prop.~\ref{prop:postbpp:postbqp}.

\begin{proposition}
  \label{prop:postbpp:bqp}
        {\bf PostBPP} does not include {\bf BQP}.  This proposition is
        based on the conjecture that {\bf BQP} is not included in {\bf
          PH}.
\end{proposition}

This means that there is at least one language $L$ that belongs to
{\bf BQP} but not to {\bf PostBPP}.

In general, proving that a class is not included in another one is
difficult. In particular, proving that {\bf BQP} is not included in
other classical classes is still difficult, even based on conjectures
on classical computation classes, like in the proof of
Prop.~\ref{prop:postbpp:postbqp}. Such statements are only known
relative to oracles, see, e.g., Ref.~\cite{raz2018}.

It is widely believed that {\bf BQP} contains problems outside {\bf
  NP}.  The recent results of Ref.~\cite{raz2018} try to compare {\bf
  BQP} with {\bf PH}, which includes {\bf NP} and is conjectured to be
separated from it. The result is however only based on the oracle
version of the classes. Althouhgh the theorem informally suggests that
{\bf BQP} has abilities that are outside {\bf PH}, it still does not
represent a proof.

I will give a proof of Prop.~\ref{prop:postbpp:bqp} based on the
conjecture that {\bf BQP} is not included in {\bf PH}. Unfortunately,
this conjecture is much less supported than the one used for proving
Prop.~\ref{prop:postbpp:postbqp}, i.e. that the polynomial hierarchy
does not collapse. It is worth noting that the proofs of the following
sections will not necessarily rely on this (less reliable) result; it
is only given as a possible alternative reasoning.

\begin{proof}
  This proof is based on the conjecture that {\bf BQP} is not included
  in {\bf PH}.

  Ref.~\cite{bravyi2007} shows that {\bf PostBPP} equals
  {\bf BPP$_{path}$}, defined in Ref.~\cite{han1997}. This class is
  included in {\bf BPP$^{\mathrm{NP}}$} (Ref.~\cite{han1997}), which,
  in turn, is included in {\bf $\Sigma^p_3$} (see Ref.~\cite{han1997},
  Fig. 1; the inclusion is considered trivial) and is thus part of
  {\bf PH} (see the definition of {\bf PH}).

  We thus find that {\bf PostBPP} is included in {\bf PH}. Under the
  conjecture that {\bf BQP} is not included in {\bf PH}, we prove the
  thesis.
\end{proof}


\section{Proof of propositions on sampling complexity classes}

\label{sect:no:go:sampling:post:selection}

These proofs are based on results of Appendix~\ref{sect:decision}.

\subsection{Proof of Prop.~\ref{prop:sampp:not:circuit}}

\begin{proof}
  According to the conjecture, there is a language $L$ that belongs to
  {\bf BQP} but not to {\bf BPP}. According to the definition of {\bf
    BQP}, that language $L$ can be decided using a polynomial-time
  uniform family of quantum circuits $\left\{Q_n : n\in \mathbb{N}
  \right\}$ with one bit of output. Let us call $F$ the problem of
  sampling from this family of quantum circuits. I show that the
  problem $F$ is not in {\bf SampP}, i.e. the quantum circuits of this
  family cannot be sampled by the algorithms used in {\bf SampP}.

  Due to the origin of $Q_n$ from the definition of the language $L$
  in {\bf BQP}, the probability distribution $\mathcal{D}_{Q_{|x|},x}(y)$ of
  $F$ is, for $x\in L$:
  \begin{equation}
    \mathcal{D}_{Q_{|x|},x}(y=0) < 1/3 \; , \; \mathcal{D}_{Q_{|x|},x}(y=1) > 2/3
  \end{equation}
  and, for $x\notin L$:
  \begin{equation}
    \mathcal{D}_{Q_{|x|},x}(y=0) > 2/3 \; , \; \mathcal{D}_{Q_{|x|},x}(y=1) < 1/3
  \end{equation}
  
  By contradiction, I assume that $F$ belongs to {\bf
    SampP}. Then, there is an algorithm $B$ that generates samples $y$
  as a function of $x$ with a probability $\mathcal{C}_{Q_{|x|},x}(y)$:
  \begin{equation}
    \sum_{y\in \{0,1\}^{p\left(|x|\right)}}
    \left| \mathcal{C}_{Q_{|x|},x}(y) - \mathcal{D}_{Q_{|x|},x}(y) \right| \le \epsilon
  \end{equation}
  For each one-bit output $y$:
  \begin{equation}
    \mathcal{D}_{Q_{|x|},x}(y) - \epsilon
    \le \mathcal{C}_{Q_{|x|},x}(y)
    \le \mathcal{D}_{Q_{|x|},x}(y) + \epsilon
  \end{equation}
  Using the suitable inequality, chosen between these two, we get, for $x\in L$:
  \begin{equation}
    \mathcal{C}_{Q_{|x|},x}(y=0)  \le \mathcal{D}_{Q_{|x|},x}(y=0) + \epsilon < 1/3 + \epsilon
    \; , \;
    \mathcal{C}_{Q_{|x|},x}(y=1)  \ge \mathcal{D}_{Q_{|x|},x}(y=1) - \epsilon > 2/3 -\epsilon
  \end{equation}
  and, for  $x\notin L$:
  \begin{equation}
    \mathcal{C}_{Q_{|x|},x}(y=0)  \ge \mathcal{D}_{Q_{|x|},x}(y=0) - \epsilon > 2/3 -\epsilon
    \; , \;
    \mathcal{C}_{Q_{|x|},x}(y=1)  \le \mathcal{D}_{Q_{|x|},x}(y=1) + \epsilon < 1/3 + \epsilon
  \end{equation}
  This algorithm $B$ can thus be used in the definition of {\bf BPP},
  leading to an error $1/3+\epsilon$. If $\epsilon$ is small enough,
  $1/3+\epsilon<1/2$, the class remains the same. Then we conclude
  that $L$ belongs to {\bf BPP}, which is a contradiction.
\end{proof}

\subsection{Proof of Prop.~\ref{prop:postsampp:not:circuit}}

\begin{proof}
  This proof is based on Prop.~\ref{prop:postbpp:bqp}, hence it
  relies on the conjecture that {\bf BQP} is not included
  in {\bf PH}.
  
  Based on Prop.~\ref{prop:postbpp:bqp}, we know that there is a
  language $L$ that belongs to {\bf BQP} but not to {\bf
    PostBPP}. According to the definition of {\bf BQP}, that language
  $L$ can be decided using a polynomial-time uniform family of quantum
  circuits $\left\{Q_n : n\in \mathbb{N} \right\}$ with one bit of
  output. I call $F$ the problem of sampling from this family of
  quantum circuits; I show that the quantum circuits of this family
  cannot be sampled by the algorithms used in {\bf PostSampP}.

  Due to the origin of $Q_n$ from the definition of the language $L$
  in {\bf BQP}, the probability distribution $\mathcal{D}_{Q_{|x|},x}(y)$ of
  $F$ is, for $x\in L$:
  \begin{equation}
    \mathcal{D}_{Q_{|x|},x}(y=0) < 1/3 \; , \; \mathcal{D}_{Q_{|x|},x}(y=1) > 2/3
  \end{equation}
  and, for $x\notin L$:
  \begin{equation}
    \mathcal{D}_{Q_{|x|},x}(y=0) > 2/3 \; , \; \mathcal{D}_{Q_{|x|},x}(y=1) < 1/3
  \end{equation}
  
  By contradiction, I assume that $F$ belongs
  to {\bf PostSampP}. Then, there is an algorithm $B$, with the
  possibility of returning 'FAILED', which generates samples $y$ as a
  function of $x$ with a probability $\mathcal{C}_{Q_{|x|},x}(y)$ (conditioned
  to have a valid output, i.e. not 'FAILED'):
  \begin{equation}
    \sum_{y\in \{0,1\}^{p\left(|x|\right)}}
    \left| \mathcal{C}_{Q_{|x|},x}(y) - \mathcal{D}_{Q_{|x|},x}(y) \right| \le \epsilon
  \end{equation}
  For each one-bit output $y$:
  \begin{equation}
    \mathcal{D}_{Q_{|x|},x}(y) - \epsilon
    \le \mathcal{C}_{Q_{|x|},x}(y)
    \le \mathcal{D}_{Q_{|x|},x}(y) + \epsilon
  \end{equation}
  Using the suitable inequality, chosen between these two, we get, for $x\in L$:
  \begin{equation}
    \mathcal{C}_{Q_{|x|},x}(y=0)  \le \mathcal{D}_{Q_{|x|},x}(y=0) + \epsilon < 1/3 + \epsilon
    \; , \;
    \mathcal{C}_{Q_{|x|},x}(y=1)  \ge \mathcal{D}_{Q_{|x|},x}(y=1) - \epsilon > 2/3 -\epsilon
  \end{equation}
  and, for  $x\notin L$:
  \begin{equation}
    \mathcal{C}_{Q_{|x|},x}(y=0)  \ge \mathcal{D}_{Q_{|x|},x}(y=0) - \epsilon > 2/3 -\epsilon
    \; , \;
    \mathcal{C}_{Q_{|x|},x}(y=1)  \le \mathcal{D}_{Q_{|x|},x}(y=1) + \epsilon < 1/3 + \epsilon
  \end{equation}
  This algorithm $B$ can thus be used in the definition of {\bf PostBPP},
  leading to an error $1/3+\epsilon$. If $\epsilon$ is small enough,
  $1/3+\epsilon<1/2$, the class remains the same. Then we conclude
  that $L$ belongs to {\bf PostBPP}, which is a contradiction.
\end{proof}

\subsection{Proof of Prop.~\ref{prop:postsampp:star:not:circuit}}

In order to proceed with the proof, I first give a definition.
\begin{definition}[Conditioned probability]
  Given a $\mathcal{D}_x(y)$, with $|y|=p\left(|x|\right)$, we define
  the conditioned probability $\mathcal{\bar{D}}_x(k, b, y')$
  representing the probability of getting $y_k=b$ for a given index
  $k$ and bit value $b$, conditioned to all the other values $j\ne k$
  to be $y_j=y'_j$, where $y'\in \{0,1\}^{p(|x|)}$ is a given string:
  \begin{equation}
    \mathcal{\bar{D}}_x(k, b, y') =
    \frac{
      \mathcal{D}_x(y_k=b, y_j=y'_j)
    }{
      \mathcal{D}_x(y_k=0, y_j=y'_j) + \mathcal{D}_x(y_k=1, y_j=y'_j)
    }
  \end{equation}
\end{definition}

Moreover, we need the following lemma.

\begin{lemma}
  \label{lemma:multiplicative:error:conditioned:probability}
  Given a probability distributions $\mathcal{D}_x(y)$ and an $\epsilon'>0$,
  it is possible to find an $\epsilon>0$ such that, for every
  probability distribution $\mathcal{C}_x(y)$ that satisfies
  \begin{equation}
    \label{eq:lemma:epsilon}
    \left| \frac{\mathcal{C}_x(y)}{\mathcal{D}_x(y)} -1 \right| \le \epsilon
  \end{equation}
  for every $y$, it also satisfies
  \begin{equation}
    \label{eq:lemma:epsilon:p}
    \left| \mathcal{\bar{C}}_x(k, b, y') - \mathcal{\bar{D}}_x(k, b,
    y') \right| \le \epsilon'
  \end{equation}
  for every $k$, $0\le k<p\left(|x|\right)$, $b=0$ or 1, and $y'\in
  \{0,1\}^{p(|x|)}$.
\end{lemma}

\begin{proof}
  Given $k$ and $y'$, I introduce the following shortcuts:
  \begin{eqnarray}
    a & = & \mathcal{C}_x(y_k=0, y_j=y'_j) \\
    b & = & \mathcal{C}_x(y_k=1, y_j=y'_j) \\
    A & = & \mathcal{D}_x(y_k=0, y_j=y'_j) \\
    B & = & \mathcal{D}_x(y_k=1, y_j=y'_j)
  \end{eqnarray}
  Equation~\ref{eq:lemma:epsilon} lead to:
  \begin{eqnarray}
    A \left(1-\epsilon\right) & \le a \le & A \left(1+\epsilon\right) \\
    B \left(1-\epsilon\right) & \le b \le & B \left(1+\epsilon\right)     
  \end{eqnarray}
  Now we evaluate the conditioned probability:
  \begin{equation}
    \mathcal{\bar{C}}_x(k, 0, y') = \frac{a}{a+b} \le
    \frac{A\left(1+\epsilon\right)}{A\left(1-\epsilon\right)+B\left(1-\epsilon\right)} \le \mathcal{\bar{D}}_x(k, 0, y') + \frac{2\epsilon}{1-\epsilon}
  \end{equation}
  Using an analogous procedure, we find:
  \begin{eqnarray}
    \mathcal{\bar{D}}_x(k, 0, y') - \frac{2\epsilon}{1+\epsilon} 
    & \le \mathcal{\bar{C}}_x(k, 0, y') \le &
    \mathcal{\bar{D}}_x(k, 0, y') + \frac{2\epsilon}{1-\epsilon} \\
    \mathcal{\bar{D}}_x(k, 1, y') - \frac{2\epsilon}{1+\epsilon}
    & \le \mathcal{\bar{C}}_x(k, 1, y') \le &
    \mathcal{\bar{D}}_x(k, 1, y') + \frac{2\epsilon}{1-\epsilon}
  \end{eqnarray}
  Given $\epsilon'$, we choose $\epsilon$ such that: 
  \begin{equation}
    \epsilon' = \frac{2\epsilon}{1-\epsilon}
  \end{equation}
  For this $\epsilon'$:
  \begin{equation}
    \epsilon' \ge \frac{2\epsilon}{1+\epsilon}
  \end{equation}
  We get the inequalities:
  \begin{eqnarray}
    \mathcal{\bar{D}}_x(k, 0, y') - \epsilon' 
    & \le \mathcal{\bar{C}}_x(k, 0, y') \le &
    \mathcal{\bar{D}}_x(k, 0, y') + \epsilon' \\
    \mathcal{\bar{D}}_x(k, 1, y') - \epsilon'
    & \le \mathcal{\bar{C}}_x(k, 1, y') \le &
    \mathcal{\bar{D}}_x(k, 1, y') + \epsilon'
  \end{eqnarray}
  which are equivalent to the thesis.  
\end{proof}

We can now prove the  Prop.~\ref{prop:postsampp:star:not:circuit}.
  
\begin{proof}
  This proof is based on Prop.~\ref{prop:postbpp:postbqp}, hence it
  relies on the conjecture that the polynomial hierarchy
  does not collapse.

  Based on Prop.~\ref{prop:postbpp:postbqp}, we know that there is a
  language $L$ that belongs to {\bf PostBQP} but not to {\bf
    PostBPP}. According to the definition of {\bf PostBQP}, we decide
  that language $L$ using a polynomial-time uniform family of quantum
  circuits $\left\{Q_n : n\in \mathbb{N} \right\}$ with two bits of
  output.  I call $F$ the problem of sampling from this family of quantum
  circuits; I
  show that the quantum circuits of this family cannot be sampled by
  the algorithms used in {\bf PostSampP*}.

  Due to the origin of $Q_n$ from the definition of the language $L$
  in {\bf PostBQP}, the probability distribution $\mathcal{D}_{Q_{|x|},x}(y)$ of
  $F$ is, for $x\in L$:
  \begin{equation}
    \frac{\mathcal{D}_{Q_{|x|},x}(y_{\mathrm{sample}}=0, y_{\mathrm{valid}}=1)} {\mathcal{N}} < 1/3
    \; , \;
    \frac{\mathcal{D}_{Q_{|x|},x}(y_{\mathrm{sample}}=1, y_{\mathrm{valid}}=1)} {\mathcal{N}} > 2/3
  \end{equation}
  and, for $x\notin L$:
  \begin{equation}
    \frac{\mathcal{D}_{Q_{|x|},x}(y_{\mathrm{sample}}=0, y_{\mathrm{valid}}=1)} {\mathcal{N}} > 2/3
    \; , \;
    \frac{\mathcal{D}_{Q_{|x|},x}(y_{\mathrm{sample}}=1, y_{\mathrm{valid}}=1)} {\mathcal{N}} < 1/3
  \end{equation}
  where:
  \begin{equation}
    \mathcal{N} = \mathcal{D}_{Q_{|x|},x}(y_{\mathrm{sample}}=0, y_{\mathrm{valid}}=1) +
    \mathcal{D}_{Q_{|x|},x}(y_{\mathrm{sample}}=1, y_{\mathrm{valid}}=1)
  \end{equation}
  and $\mathcal{N}$ does not vanish for any $Q_n$ and $x$.

  We can rewrite these inequalities using the conditioned probability
  $\mathcal{\bar{D}}_{Q_{|x|},x}(k, b, y')$. Defining
  $y'=\left<0,1\right>$, for $x\notin L$:
  \begin{equation}
    \mathcal{\bar{D}}_{Q_{|x|},x}(0, 0, y') < 1/3
    \; , \;
    \mathcal{\bar{D}}_{Q_{|x|},x}(0, 1, y') > 2/3
  \end{equation}
  and, for $x\notin L$:
  \begin{equation}
    \mathcal{\bar{D}}_{Q_{|x|},x}(0, 0, y') > 2/3
    \; , \;
    \mathcal{\bar{D}}_{Q_{|x|},x}(0, 1, y') < 1/3
  \end{equation} 
  
  By contradiction, I assume that $F$ belongs to {\bf
    PostSampP*}. Then, there is an algorithm $B$, with the possibility
  of returning 'FAILED', which generates samples $y$ as a function of
  $x$ with a probability $\mathcal{C}_{Q_{|x|},x}(y)$ (conditioned to
  have a valid output, i.e. not 'FAILED'); the conditioned probability
  is $\mathcal{\bar{C}}_{Q_{|x|},x}(k, b, y')$. The condition on the
  error is:
  \begin{equation}
    \left| \frac{\mathcal{C}_{Q_{|x|},x}(y)}{\mathcal{D}_{Q_{|x|},x}(y)} -1 \right| \le \epsilon
  \end{equation}
  for every $y$. Thanks to Lemma~\ref{lemma:multiplicative:error:conditioned:probability}:
  \begin{equation}
    \left| \mathcal{\bar{C}}_{Q_{|x|},x}(k, b, y') - \mathcal{\bar{D}}_{Q_{|x|},x}(k, b,
    y') \right| \le \epsilon
  \end{equation}
  for every $k$, $0\le k<p\left(|x|\right)$, $b=0$ or 1, and $y'\in
  \{0,1\}^{p(|x|)}$.
   
  For $y'=\left<0,1\right>$ and for each $y_{\mathrm{sample}}=b$ bit:
   \begin{equation}
    \mathcal{\bar{D}}_{Q_{|x|},x}(0, b, y') - \epsilon
    \le \mathcal{\bar{C}}_{Q_{|x|},x}(0, b, y')
    \le \mathcal{\bar{D}}_{Q_{|x|},x}(0, b, y') + \epsilon
  \end{equation}
  Using the suitable inequality, chosen between these two, we get, for $x\in L$:
  \begin{equation}
    \mathcal{\bar{C}}_{Q_{|x|},x}(0, 0, y')  \le \mathcal{\bar{D}}_{Q_{|x|},x}(0, 0, y') + \epsilon < 1/3 + \epsilon
    \; , \;
    \mathcal{\bar{C}}_{Q_{|x|},x}(0, 1, y')  \ge \mathcal{\bar{D}}_{Q_{|x|},x}(0, 1, y') - \epsilon > 2/3 -\epsilon
  \end{equation}
  and, for  $x\notin L$:
  \begin{equation}
    \mathcal{\bar{C}}_{Q_{|x|},x}(0, 0, y')  \ge \mathcal{\bar{D}}_{Q_{|x|},x}(0, 0, y') - \epsilon > 2/3 -\epsilon
    \; , \;
    \mathcal{\bar{C}}_{Q_{|x|},x}(0, 1, y')  \le \mathcal{\bar{D}}_{Q_{|x|},x}(0, 1, y') + \epsilon < 1/3 + \epsilon
  \end{equation}
  We tweak this algorithm $B$, obtaining $M$, as follows. If $B$
  returns 'FAILED', $M$ also returns 'FAILED'.  Else, $B$ returns two
  bits, $y_{\mathrm{sample}}$, $y_{\mathrm{valid}}$. Then $M$
  inspects $y_{\mathrm{valid}}$: if it is 0, then it returns 'FAILED',
  else, it returns $y_{\mathrm{sample}}$. The latter will take place
  with non-vanishing probability, bacause $B$ does not return 'FAILED'
  with non-vanishing probability and $\mathcal{N}$ does not vanish for
  any $x$. Moreover, conditioned to not returning 'FAILED', the bit
  returned by $M$ has the probability distribution, for $x\in L$:
  \begin{equation}
    \mathcal{P}_{Q_{|x|},x}(0) < 1/3 + \epsilon \; , \; \mathcal{P}_{Q_{|x|},x}(1) > 2/3 -\epsilon
  \end{equation}
  and, for  $x\notin L$:
  \begin{equation}
    \mathcal{P}_{Q_{|x|},x}(0) > 2/3 -\epsilon
    \; , \;
    \mathcal{P}_{Q_{|x|},x}(1) < 1/3 + \epsilon
  \end{equation}
  We thus see that $M$ matches the requirements for the algorithm $M$
  in the definition of {\bf PostBPP}, leading to an error
  $1/3+\epsilon$. If $\epsilon$ is small enough, $1/3+\epsilon<1/2$,
  the class remains the same. Then we conclude that $L$ belongs to
  {\bf PostBPP}, which is a contradiction.
\end{proof}

\bibliographystyle{unsrt}
\bibliography{hidden-variable}

\end{document}